\documentclass[a4paper,12pt]{article}
\usepackage[dvipdfmx]{graphicx,color}
\usepackage{amssymb}
\usepackage{amsmath}
\usepackage{latexsym}
\usepackage{array}
\usepackage{bm}
\usepackage{ascmac}
\usepackage{color}
\usepackage{ulem}

\def\Z{\mathbb{Z}}

\def\qed{\hfill\hbox{\rule[-2pt]{3pt}{6pt}}}

\newtheorem{thm}{Theorem}[section]

\newtheorem{prop}[thm]{Proposition}

\newtheorem{dfn}[thm]{Definition}

\newtheorem{proof}{Proof}

\makeatletter

\@addtoreset{equation}{section}
\makeatother

\title{Laurent phenomenon algebras and the discrete BKP equation}
\author{Naoto Okubo\\Graduate School of Mathematical Sciences, the University of Tokyo, \\3-8-1 Komaba, Tokyo 153-8914, Japan}
\date{}

\begin{document}
\maketitle

\begin{abstract}
We construct the Laurent phenomenon algebras the cluster variables of which satisfy the discrete BKP equation and other difference equations obtained by its reduction.
These Laurent phenomenon algebras are constructed from seeds with a generalization of mutation-period property.
We show that a reduction of a seed corresponds to a reduction of a difference equation.
\end{abstract}

\section{Introduction}

In this article, we deal with the Laurent phenomenon algebras introduced by Lam and Pylyavskyy \cite{LP}.
The Laurent phenomenon algebra is a generalization of the cluster algebra \cite{cluster} with Laurent phenomenon of cluster variables.
A Laurent phenomenon algebra is a commutative ring described by cluster variables.
A generating set of the Laurent phenomenon algebra is defined by mutation, which is a transformation of a seed consisting of a set of cluster variables and exchange polynomials.
The Laurent phenomenon is the property that all cluster variables obtained by mutations of an initial seed are Laurent polynomials in cluster variables of the initial seed.
It is known that cluster variables of suitable cluster algebras can satisfy the Hirota–Miwa equation \cite{BKP}, the discrete KdV equation \cite{dKdV}, the discrete Toda equation \cite{dToda}, Soms-4, and Somos-5 \cite{Somos}, when the initial seed includes suitable periodic quivers \cite{quiver1}.
Difference equations obtained by a cluster algebra are restricted to the form:
\begin{equation}\label{eq:cluster}
x'_kx_k=\prod_{i\in I}x_i^{a_i}+\prod_{j\in J}x_j^{b_j}.
\end{equation}
However, the discrete BKP equation
\begin{equation}
x_{n+1}^{m+1,l+1}x_{n}^{m,l}
=x_{n}^{m+1,l+1}x_{n+1}^{m,l}
+x_{n+1}^{m+1,l}x_{n}^{m,l+1}
+x_{n+1}^{m,l+1}x_{n}^{m+1,l}
\end{equation}
does not have the form given by \eqref{eq:cluster}.
Cluster variables of a Laurent phenomenon algebra obtained from a mutation of an initial seed can satisfy the difference equations of more general form as
\begin{equation}
x'_kx_k=P\in\mathbb{Z}[\left.x_i\right|i\in I].
\end{equation}
It is known that cluster variables of a Laurent phenomenon algebra can satisfy the Somos-6 \cite{Somos} and some related difference equations, when the initial seed includes suitable exchange polynomials \cite{LP,period}.
These seeds have the property called ‘mutation-period’.
This property is an analogue of the mutation-periodic quiver of cluster algebras.
Several results concerning period-1 seeds have been reported in \cite{period}.
In this paper, we construct the Laurent phenomenon algebras whose cluster variables satisfy the discrete BKP equation \cite{BKP}, the Somos-7, and several other difference equations.
The seed which gives the discrete BKP equation has infinite rank.
This seed has the property of a generalization of the mutation-period.
The seeds which give the Somos-6, Somos-7, and related 2-dimensional difference equations are obtained from reductions of the seed that gives the discrete BKP equation.
This reduction corresponds to a reduction from the discrete BKP equation to these difference equations.

\section{Laurent phenomenon algebras}

In this section, we briefly explain the notion of Laurent phenomenon algebra which we use in the following sections.

\subsection{Definition of Laurent phenomenon algebra}

Let $\bm{x}=(x_0,x_1,\dots ,x_{N-1})$ be an $N$-tuple of variables.
Let $\bm{F}=(F_0,F_1,\dots ,F_{N-1})$ be an $N$-tuple of polynomials in $\mathbb{Z}[x_0,x_1,\dots,x_{N-1}]$.
We assume that these polynomials satisfy the following conditions:
\begin{itemize}
\item (LP1) $F_i\in\mathbb{Z}[x_0,x_1,\dots,x_{N-1}]$ is irreducible and is not divisible by any $x_j$.

\item (LP2) $F_i$ does not depend on $x_i$.
\end{itemize}
Each $x_i$ is called a cluster variable and each $F_i$ is called an exchange polynomial.
The pair $t=(\bm{x},\bm{F})$ is called a seed.
We will often express a seed as
\begin{equation}
t=\{(x_0,F_0),(x_1,F_1),\dots,(x_{N-1},F_{N-1})\}.
\end{equation}
The number of cluster variables $N$ in a seed $t$ is called the rank of the seed $t$ or the rank of the Laurent phenomenon algebra.
For Laurent polynomials $F,G\in\mathbb{Z}[x_0^{\pm},x_1^{\pm},\dots,x_{N-1}^{\pm}]$ and a cluster variable $x_i$, let $\left.F\right|_{x_i\gets G}$ be the polynomial in which we substituted $x_i$ for $G$ in $F$.
A mutation is a particular transformation of seeds.

\begin{dfn}{\rm\cite{LP}}
The mutation of the seed $t$ at $x_k$ is a transformation from $t$ to $t'$, 
\begin{equation}
\mu_k: \; t=\{(x_0,F_0),(x_1,F_1),\dots,(x_{N-1},F_{N-1})\}\,\mapsto \, t'=\{(x'_0,F'_0),(x'_1,F'_1),\dots,(x'_{N-1},F'_{N-1})\},
\end{equation}
that are defined via the following sequence of steps:
\begin{enumerate}
\item Let $\hat{F}_k\in\mathbb{Z}\left[x_0^\pm,x_1^\pm,\dots,x_{k-1}^\pm,x_{k+1}^\pm,\dots,x_{N-1}^\pm\right]$ be the unique Laurent polynomial satisfying following conditions:
\begin{itemize}
\item There exist $a_0,a_1,\dots,a_{N-1}\in\mathbb{Z}_{\leq0}$ such that
\begin{equation}
\hat{F}_k=x_0^{a_0}x_1^{a_1}\dots x_{k-1}^{a_{k-1}}x_{k+1}^{a_{k+1}}\dots x_{N-1}^{a_{N-1}}F_k.
\end{equation}

\item For any $j\neq k$,
\begin{equation}
\left.\hat{F}_k\right|_{x_j\gets F_j/x_j}\in\mathbb{Z}\left[x_0^\pm,x_1^\pm,\dots,x_{j-1}^\pm,x_j^\pm,x_{j+1}^\pm,\dots,x_{N-1}^\pm\right]
\end{equation}
and is not divisible by $F_j$.
\end{itemize}

\item New cluster variable $x'_i$ is defined as:
\begin{equation}
\begin{aligned}
x'_i&=x_i\quad(i\neq k),\\
x'_k&=\hat{F}_k/x_k.
\end{aligned}
\end{equation}

\item If $F_i$ does not depend on $x_k$, then we define $F'_i=F_i$.
In the following steps, we assume that $F_i$ depends on $x_k$.

\item We define the Laurent polynomial $G_i$ by
\begin{equation}
G_i=F_i\left|_{x_k\gets\frac{\left.\hat{F}_k\right|_{x_i\gets0}}{x'_k}}\right..
\end{equation}
Note that we can show that if $F_i$ depends on $x_k$, then $a_i=0$.
Therefore, $\left.\hat{F}_k\right|_{x_i\gets0}$ is well defined.

\item We define $H_i$ to be the result of removing all common factors in 

$\mathbb{Z}[x_0,x_1,\dots,x_{k-1},x_{k+1},\dots,x_{i-1},x_{i+1},\dots,x_{N-1}]$ with $\left.\hat{F}_k\right|_{x_i\gets0}$ from $G_i$.

\item Let $M_i$ be the unique Laurent monomial of $x'_0,x'_1,\dots,x'_{N-1}$ satisfying following conditions:
\begin{itemize}
\item $M_iH_i\in\mathbb{Z}[x'_0,x'_1,\dots,x'_{N-1}]$.

\item $M_iH_i$ is not divisible by any $x'_j$.

\item $M_iH_i$ does not depend on $x'_i$.
\end{itemize}
New exchange polynomial $F'_i$ is defined as $F'_i=M_iH_i$.
\end{enumerate}
\end{dfn}

\begin{prop}{\rm \cite{LP}}
For any $k\in\{0,1,\dots,N-1\}$ and any seed $t$, it holds that $\mu_k^2(t)=t$.
\end{prop}

\begin{prop}{\rm \cite{LP}}
Suppose that we mutate at $x_k$.
Then $F'_i$ depends on $x'_k$ if and only if $F_i$ depends on $x_k$.
\end{prop}

\begin{dfn}{\rm \cite{LP}}
Let us fix a seed $t=(\bm{x},\bm{F})$.
This seed is called an initial seed.
Let $X(t)$ be the set of all the cluster variables obtained by iterative mutations to the initial seed $t$.
Let $\mathcal{A}(t)$ be a Laurent phenomenon algebra defined as
\begin{equation}
\mathcal{A}(t)=\mathbb{Q}[x\in X(t)]\subset\mathbb{Q}(\bm{x}).
\end{equation}
\end{dfn}

\begin{thm}{\rm \cite{LP}}\label{Laurent}
Let $t=(\bm{x},\bm{F})$ be a initial seed.
If $x\in X(t)$, then $x\in\mathbb{Z}[\bm{x}^\pm]$.
\end{thm}

\subsection{Period-1 seeds}

For a polynomial
\begin{equation}
f(x_0,x_1,\dots,x_{N-1})\in\mathbb{Z}[x_0,x_1,\dots,x_{N-1}]
\end{equation}
and for $j\geq 1$, we define
\begin{equation}
u^j(f)=f(x_j,x_{j+1},\dots,x_{j+N-1})\in\mathbb{Z}[x_j,x_{j+1},\dots,x_{j+N-1}].
\end{equation}
We also abbreviate $u:=u^1$.

\begin{dfn}{\rm \cite{period}}\label{period}
A seed
\begin{equation}
t=\{(x_0,F_0),(x_1,F_1),\dots,(x_{N-1},F_{N-1})\}
\end{equation}
is called a period-1 seed if 
\begin{equation}
\mu_0(t)=\{(x_N,u(F_{N-1})),(x_1,u(F_0)),(x_2,u(F_1)),\dots,(x_{N-1},u(F_{N-2}))\}
\end{equation}
holds, where $x_N=x'_0$ is the new cluster variable obtained by the mutation at $x_0$.
\end{dfn}

We assume that
\begin{equation}
t_0=\{(x_0,F_0),(x_1,F_1),\dots,(x_{N-1},F_{N-1})\}
\end{equation}
is a period-1 seed.
We inductively define the seed $t_n$ by $t_n=\mu_{n-1}(t_{n-1})$.
The  new cluster variable $x_n'$ in $\mu_n(t_n)$ is expressed as $x_{n+N}$. Hereafter, we write this process as $x_n \to x_{n+N}$.
We have
\begin{equation}
t_n=\{(x_n,u^n(F_0)),(x_{n+1},u^n(F_1)),\dots,(x_{n+N-1},u^n(F_{N-1}))\}.
\end{equation}
If an exchange polynomial of a period-1 seed satisfies $\hat{F}_0=F_0$, then the cluster variables satisfy $x_{n+N}x_n=u^n(F_0)$.
In the following sections, all exchange polynomials satisfy $\hat{F}_i=F_i$.
We put $F_0=f(x_1,x_2,\dots,x_{N-1})$.
All cluster variables satisfy the difference equation
\begin{equation}
x_{n+N}x_n=f(x_{n+1},x_{n+2},\dots,x_{n+N-1}).
\end{equation}

Several results about a period-1 seed have already known.
\begin{itemize}
\item All the rank 2 or 3 period-1 seeds have been obtained \cite{period}.

\item An example of a rank 6 period-1 seed has been obtained \cite{LP}.
The cluster variables of this seed satisfy the difference equation
\begin{equation}
x_{n+6}x_n=x_{n+5}x_{n+1}+x_{n+4}x_{n+2}+x_{n+3}^2.
\end{equation}
This difference equation is called Somos-6.

\item Several rank $N$ period-1 seeds have been obtained \cite{period}.
For example, the cluster variables of these seeds satisfy the following difference equations:
\begin{equation}
x_{n+N}x_n
=\prod_{i=1}^{N-1}x_{n+i}^{a_i}+1\quad(a_i\in\mathbb{Z}_{>0}),
\end{equation}
\begin{equation}
x_{n+N}x_n
=x_{n+1}x_{n+N-1}+A\sum_{i=1}^{N-1}x_{n+i}+B\quad(A,B\in\mathbb{Z}).
\end{equation}
\end{itemize}

\section{Seed of the discrete BKP equation}

\subsection{Seed of the discrete BKP equation and its mutation}

In this section, we consider the infinite rank seed.
For $i\leq j\ (i,j\in\Z)$, we put $[i,j]=\{i,i+1,\dots,j\}$.
For a polynomial $P\in\mathbb{Z}[x_n^{m,l}|n,m,l\in\mathbb{Z}]$, let $s_i^{j,k}(P)$ be the polynomial in which we substituted all variables $x_n^{m,l}$ for $x_{n+i}^{m+j,l+k}$ in $P$.
We take 
\begin{equation}\label{dBKP seed}
t_0=\left\{\left.\left(x_{i-2j-k}^{j-k,k},F_{i-2j-k}^{j-k,k}\right)\right|i\in[0,5],j,k\in\mathbb{Z}\right\}
\end{equation}
as an initial seed, where exchange polynomials are defined as
\begin{equation}\label{dBKP polynomial}
\begin{aligned}
F_{0}^{0,0}
&=x_0^{1,1}x_1^{0,0}
+x_0^{1,0}x_1^{0,1}
+x_0^{0,1}x_1^{1,0},\\
F_{1}^{0,0}
&=x_0^{1,0}x_1^{0,1}x_2^{0,0}
+x_0^{0,1}x_1^{1,0}x_2^{0,0}
+x_0^{0,0}x_1^{1,0}x_2^{0,1}
+x_0^{0,0}x_1^{0,1}x_2^{1,0},\\
F_{2}^{0,0}
&=x_1^{1,0}x_2^{-1,0}x_2^{0,1}x_3^{0,0}
+x_1^{0,1}x_2^{-1,0}x_2^{1,0}x_3^{0,0}\\
&\quad+x_1^{0,0}x_2^{-1,1}x_2^{1,0}x_3^{0,0}
+x_1^{0,0}x_2^{-1,0}x_2^{0,1}x_3^{1,0}
+x_1^{0,0}x_2^{0,1}x_2^{1,0}x_3^{-1,0},\\
F_{3}^{0,0}
&=x_4^{-1,0}x_3^{1,0}x_3^{0,-1}x_2^{0,0}
+x_4^{0,-1}x_3^{1,0}x_3^{-1,0}x_2^{0,0}\\
&\quad+x_4^{0,0}x_3^{1,-1}x_3^{-1,0}x_2^{0,0}
+x_4^{0,0}x_3^{1,0}x_3^{0,-1}x_2^{-1,0}
+x_4^{0,0}x_3^{0,-1}x_3^{-1,0}x_2^{1,0},\\
F_{4}^{0,0}
&=x_5^{-1,0}x_4^{0,-1}x_3^{0,0}
+x_5^{0,-1}x_4^{-1,0}x_3^{0,0}
+x_5^{0,0}x_4^{-1,0}x_3^{0,-1}
+x_5^{0,0}x_4^{0,-1}x_3^{-1,0},\\
F_{5}^{0,0}
&=x_5^{-1,-1}x_4^{0,0}
+x_5^{-1,0}x_4^{0,-1}
+x_5^{0,-1}x_4^{-1,0},
\end{aligned}
\end{equation}
and
\begin{equation}
F_{i-2j-k}^{j-k,k}=s_{-2j-k}^{j-k,k}\left(F_i^{0,0}\right).
\end{equation}
For example, $F_0^{1,1}$ is defined as
\begin{equation}
F_0^{1,1}
=F_{5-2\cdot2-1}^{2-1,1}
=s_{-5}^{1,1}\left(F_5^{0,0}\right)
=x_0^{0,0}x_{-1}^{1,1}+x_0^{0,1}x_{-1}^{1,0}+x_0^{1,0}x_{-1}^{0,1}.
\end{equation}
We define $\tilde{u}:=s_1^{0,0}$ and the set of cluster variables $X_i$ as
\begin{equation}
X_i:=\left\{\left.x_{i-2j-k}^{j-k,k}\right|j,k\in\mathbb{Z}\right\}.
\end{equation}
A mutation $m_i^{j,k}\ (i,j,k\in\mathbb{Z})$ denotes the mutation at $x_{i-2j-k}^{j-k,k}$.
We define the iteration of the mutations $\nu_i^n\ (i\in\mathbb{Z},n\in\mathbb{Z}_{\geq0})$ by
\begin{equation}
\begin{aligned}
\nu_i^n=\ &m_i^{n-1,-1}\circ\dots
\circ m_i^{1,-n+1}\circ m_i^{0,-n}\circ m_i^{-1,-n+1}\circ\dots
\circ m_i^{-n+1,-1}\circ m_i^{-n,0}\\
&\circ m_i^{-n+1,1}\circ\dots
\circ m_i^{-1,n-1}\circ m_i^{0,n}\circ m_i^{1,n-1}\circ\dots
\circ m_i^{n-1,1}\circ m_i^{n,0}\ (n\neq0),\\
\nu_i^0=\ &m_i^{0,0}.
\end{aligned}
\end{equation}
This $\nu_i^n$ is the iteration of the mutations at each and all $x_{i-2j-k}^{j-k,k}\in X_i\ (j+k=n)$ just once.
We define the iteration of the mutations $\tilde{\mu}_i\ (i\in\mathbb{Z})$ by
\begin{equation}
\tilde{\mu}_i=\dots\circ\nu_i^2\circ\nu_i^1\circ\nu_i^0.
\end{equation}
This $\tilde{\mu}_i$ is the iteration of the mutations at each and all $x\in X_i$ just once.
We also use the notation $x \to y$ that means a new cluster variable $x'$ is denoted by $y$, \textit{i.e.}, $y=x'$.

\begin{prop}\label{dBKP1}
We define the seed $t_1$ by $t_1=\tilde{\mu}_0(t_0)$.
When we put $x_{-2j-k}^{j-k,k}\to x_{1-2j-k}^{1+j-k,1+k}$, it holds that
\begin{equation}\label{dBKP1 seed1}
t_1=\left\{\left.\left(x_{1+i-2j-k}^{j-k,k},\tilde{u}\left(F_{i-2j-k}^{j-k,k}\right)\right)\right|i\in[0,5],j,k\in\mathbb{Z}\right\}.
\end{equation}
The new cluster variables $x_{1-2j-k}^{1+j-k,1+k}$ satisfy
\begin{equation}\label{eq}
x_{1-2j-k}^{1+j-k,1+k}x_{-2j-k}^{j-k,k}=F_{-2j-k}^{j-k,k}.
\end{equation}
\end{prop}

\begin{proof}
If $F_n^{m,l}$ does not depend on $x_{-2j-k}^{j-k,k}$, then it holds $\left(F_n^{m,l}\right)'=F_n^{m,l}$, when we mutate at $x_{-2j-k}^{j-k,k}\in X_0$.
The dependence on $x\in X_0$ of the exchange polynomials is as follows:
\begin{equation}\label{depend}
\begin{aligned}
&F_{-2j-k}^{j-k,k}\ \mbox{does not depend on}\ x\in X_0.\\
&F_{1-2j-k}^{j-k,k}\ \mbox{depends on only}\ x_{-2j-k}^{j-k,k}\in X_0.\\
&F_{2-2j-k}^{j-k,k}\ \mbox{depends on only}\ x_{-2(j-1)-k}^{(j-1)-k,k}\in X_0.\\
&F_{3-2j-k}^{j-k,k}\ \mbox{depends on only}\ x_{-2(j-1)-(k-1)}^{(j-1)-(k-1),k-1},x_{-2(j-1)-k}^{(j-1)-k,k}\in X_0.\\
&F_{4-2j-k}^{j-k,k}\ \mbox{depends on only}\ x_{-2(j-1)-(k-1)}^{(j-1)-(k-1),k-1}\in X_0.\\
&F_{5-2j-k}^{j-k,k}\ \mbox{depends on only}\ x_{-2(j-2)-(k-1)}^{(j-2)-(k-1),k-1}\in X_0.
\end{aligned}
\end{equation}
For each $i\in[0,5]$, we consider the change of $\left(x_{i-2j-k}^{j-k,k},F_{i-2j-k}^{j-k,k}\right)$ by the mutations of initial seed $t_0$ at each and all $x\in X_0$ just once.
\begin{itemize}

\item \underline{The case of $i=0$}

By \eqref{depend}, the exchange polynomials $F_{-2j-k}^{j-k,k}$ do not change by the mutations at $x\in X_0$.
For each seed
\begin{equation}
\begin{aligned}
\tilde{t}
&=\left\{\left.\left(\tilde{x}_{i-2j-k}^{j-k,k},\tilde{F}_{i-2j-k}^{j-k,k}\right)\right|i\in[0,5],j,k\in\mathbb{Z}\right\}\\
&=\left\{\left.\left(\tilde{x}_{-2j-k}^{j-k,k},F_{-2j-k}^{j-k,k}\right),\left(\tilde{x}_{i-2j-k}^{j-k,k},\tilde{F}_{i-2j-k}^{j-k,k}\right)\right|i\in[1,5],j,k\in\mathbb{Z}\right\}
\end{aligned}
\end{equation}
obtained by the mutations at $x\in X_0$, if $\left.F_{-2j-k}^{j-k,k}\right|_{\tilde{x}_n^{m,l}\gets \tilde{F}_n^{m,l}/x}$ is not divisible by $\tilde{F}_n^{m,l}$ for any $(n,m,l)\neq(-2j-k,j-k,k)$, then it holds that $\hat{F}_{-2j-k}^{j-k,k}=F_{-2j-k}^{j-k,k}$.
In fact, all exchange polynomials which appear below satisfy $\hat{F}_{-2j-k}^{j-k,k}=F_{-2j-k}^{j-k,k}$.
By the mutation at $x_{-2j-k}^{j-k,k}$, the change of the cluster variable $x_{-2j-k}^{j-k,k}$ is
\begin{equation}
x_{-2j-k}^{j-k,k}
\xrightarrow{\mu_{-2j-k}^{j-k,k}}
x_{1-2j-k}^{1+j-k,1+k}
=\hat{F}_{-2j-k}^{j-k,k}/x_{-2j-k}^{j-k,k}=F_{-2j-k}^{j-k,k}/x_{-2j-k}^{j-k,k}
\end{equation}
and other cluster variables $x_{-2j'-k'}^{j'-k',k'}\ ((j',k')\neq(j,k))$ do not change, where $\mu_n^{m,l}\ (n,m,l\in\mathbb{Z})$ is the mutation at $x_n^{m,l}$.
Therefore, we have
\begin{equation}
\begin{aligned}
\left(x_{-2j-k}^{j-k,k},F_{-2j-k}^{j-k,k}\right)
\to&\left(x_{1-2j-k}^{1+j-k,1+k},F_{-2j-k}^{j-k,k}\right)\\
&=\left(x_{6-2(j+2)-(k+1)}^{(j+2)-(k+1),k+1},\tilde{u}\left(F_{5-2(j+2)-(k+1)}^{(j+2)-(k+1),k+1}\right)\right)
\end{aligned}
\end{equation}
by the mutations of initial seed $t_0$ at each and all $x\in X_0$ just once.

\item \underline{The case of $i=1$}

By \eqref{depend}, the exchange polynomials $F_{1-2j-k}^{j-k,k}$ do not change until we execute the mutation at $x_{-2j-k}^{j-k,k}$.
Let $\lambda_1$ and $\lambda_2$ be the iteration of mutations before and after the mutation at $x_{-2j-k}^{j-k,k}$ respectively.
We have
\begin{equation}
\left(F_{1-2j-k}^{j-k,k},F_{-2j-k}^{j-k,k}\right)
\xrightarrow{\lambda_1}
\left(F_{1-2j-k}^{j-k,k},F_{-2j-k}^{j-k,k}\right).
\end{equation}
We mutate at $x_{-2j-k}^{j-k,k}$ and calculate $\left(F_{1-2j-k}^{j-k,k}\right)'$ from $F_{1-2j-k}^{j-k,k}$ and $F_{-2j-k}^{j-k,k}$.
It holds that
\begin{equation}
\left(F_{1-2j-k}^{j-k,k},F_{-2j-k}^{j-k,k}\right)
\xrightarrow{\mu_{-2j-k}^{j-k,k}}
\left(\tilde{u}\left(F_{-2j-k}^{j-k,k}\right),F_{-2j-k}^{j-k,k}\right).
\end{equation}
$\tilde{u}\left(F_{-2j-k}^{j-k,k}\right)$ does not depend on any $x\in X_0$.
We have
\begin{equation}
\tilde{u}\left(F_{-2j-k}^{j-k,k}\right)
\xrightarrow{\lambda_2}
\tilde{u}\left(F_{-2j-k}^{j-k,k}\right).
\end{equation}
The cluster variable $x_{1-2j-k}^{j-k,k}$ does not change by the mutations at $x\in X_0$.
Therefore, we have
\begin{equation}
\left(x_{1-2j-k}^{j-k,k},F_{1-2j-k}^{j-k,k}\right)
\to \left(x_{1-2j-k}^{j-k,k},\tilde{u}\left(F_{-2j-k}^{j-k,k}\right)\right)
\end{equation}
by mutations of the initial seed $t_0$ at each and all $x\in X_0$ just once.

\item \underline{The cases of $i=2,4,5$}

The cases $for$ $i=2,4,5$ are the same as the case of $i=1$.
It holds that
\begin{equation}
\left(x_{i-2j-k}^{j-k,k},F_{i-2j-k}^{j-k,k}\right)
\to \left(x_{i-2j-k}^{j-k,k},\tilde{u}\left(F_{-1+i-2j-k}^{j-k,k}\right)\right)\quad (i=2,4,5)
\end{equation}
by the mutations of initial seed $t_0$ at each and all $x\in X_0$ just once.

\item \underline{The case of $i=3$}

By \eqref{depend}, the exchange polynomials $F_{3-2j-k}^{j-k,k}$ do not change until we execute the mutation at $x_{3-2j-k}^{j-k,k-1}$ or $x_{2-2j-k}^{-1+j-k,k}$.
We consider two cases in which the order of the mutations at $x_{3-2j-k}^{j-k,k-1}$ and $x_{2-2j-k}^{-1+j-k,k}$ is different from each other.
\begin{itemize}
\item Suppose that we mutate at $x_{3-2j-k}^{j-k,k-1}$ before the mutation at $x_{2-2j-k}^{-1+j-k,k}$.
Let $\lambda_1,\lambda_2$, and $\lambda_3$ be the iteration of mutations before the mutation at $x_{3-2j-k}^{j-k,k-1}$, after the mutation at $x_{3-2j-k}^{j-k,k-1}$ and before the mutation at $x_{2-2j-k}^{-1+j-k,k}$, and after the mutation at $x_{2-2j-k}^{-1+j-k,k}$ respectively.
We have
\begin{equation}
\left(F_{3-2j-k}^{j-k,k},F_{3-2j-k}^{j-k,k-1},F_{2-2j-k}^{-1+j-k,k}\right)
\xrightarrow{\lambda_1}
\left(F_{3-2j-k}^{j-k,k},F_{3-2j-k}^{j-k,k-1},F_{2-2j-k}^{-1+j-k,k}\right).
\end{equation}
We mutate at $x_{3-2j-k}^{j-k,k-1}$ and calculate $\left(F_{3-2j-k}^{j-k,k}\right)'$ from $F_{3-2j-k}^{j-k,k}$ and $F_{3-2j-k}^{j-k,k-1}$.
Then we have
\begin{equation}
\begin{aligned}
&\left(F_{3-2j-k}^{j-k,k},F_{3-2j-k}^{j-k,k-1},F_{2-2j-k}^{-1+j-k,k}\right)\\
&\xrightarrow{\mu_{3-2j-k}^{j-k,k-1}}
\left(s_{-2j-k}^{j-k,k}(P_1),F_{3-2j-k}^{j-k,k-1},F_{2-2j-k}^{-1+j-k,k}\right),
\end{aligned}
\end{equation}
\begin{equation}
P_1
:=x_2^{-1,0}x_3^{1,0}x_4^{0,0}
+x_2^{1,0}x_3^{-1,0}x_4^{0,0}
+x_2^{0,0}x_3^{-1,0}x_4^{1,0}
+x_2^{0,0}x_3^{1,0}x_4^{-1,0}.
\end{equation}
However, $s_{-2j-k}^{j-k,k}(P_1)$ depends on only $x_{2-2j-k}^{-1+j-k,k}\in X_0$.
Hence, $s_{-2j-k}^{j-k,k}(P_1)$ does not change before the mutation at $x_{2-2j-k}^{-1+j-k,k}$.
We have
\begin{equation}
\left(s_{-2j-k}^{j-k,k}(P_1),F_{2-2j-k}^{-1+j-k,k}\right)
\xrightarrow{\lambda_2}
\left(s_{-2j-k}^{j-k,k}(P_1),F_{2-2j-k}^{-1+j-k,k}\right).
\end{equation}
Mutating at $x_{2-2j-k}^{-1+j-k,k}$ and calculating $\left(s_{-2j-k}^{j-k,k}(P_1)\right)'$ from $s_{-2j-k}^{j-k,k}(P_1)$ and $F_{2-2j-k}^{-1+j-k,k}$, we find
\begin{equation}
\left(s_{-2j-k}^{j-k,k}(P_1),F_{2-2j-k}^{-1+j-k,k}\right)
\xrightarrow{\mu_{2-2j-k}^{-1+j-k,k}}
\left(\tilde{u}\left(F_{2-2j-k}^{j-k,k}\right),F_{2-2j-k}^{-1+j-k,k}\right).
\end{equation}
Since $\tilde{u}\left(F_{2-2j-k}^{j-k,k}\right)$ does not depend on any $x\in X_0$,
we have
\begin{equation}
\tilde{u}\left(F_{2-2j-k}^{j-k,k}\right)
\xrightarrow{\lambda_3}
\tilde{u}\left(F_{2-2j-k}^{j-k,k}\right).
\end{equation}

\item Suppose that we mutate at $x_{2-2j-k}^{-1+j-k,k}$ before the mutation at $x_{3-2j-k}^{j-k,k-1}$.
Let $\lambda_1,\lambda_2$, and $\lambda_3$ be the iteration of mutations before the mutation at $x_{2-2j-k}^{-1+j-k,k}$, after the mutation at $x_{2-2j-k}^{-1+j-k,k}$ and before the mutation at $x_{3-2j-k}^{j-k,k-1}$, and after the mutation at $x_{3-2j-k}^{j-k,k-1}$ respectively.
We have
\begin{equation}
\left(F_{3-2j-k}^{j-k,k},F_{3-2j-k}^{j-k,k-1},F_{2-2j-k}^{-1+j-k,k}\right)
\xrightarrow{\lambda_1}
\left(F_{3-2j-k}^{j-k,k},F_{3-2j-k}^{j-k,k-1},F_{2-2j-k}^{-1+j-k,k}\right).
\end{equation}
By mutation at $x_{2-2j-k}^{-1+j-k,k}$ and calculation of $\left(F_{3-2j-k}^{j-k,k}\right)'$ from $F_{3-2j-k}^{j-k,k}$ and $F_{2-2j-k}^{-1+j-k,k}$, we find that
\begin{equation}
\begin{aligned}
&\left(F_{3-2j-k}^{j-k,k},F_{3-2j-k}^{j-k,k-1},F_{2-2j-k}^{-1+j-k,k}\right)\\
&\xrightarrow{\mu_{2-2j-k}^{-1+j-k,k}}
\left(s_{-2j-k}^{j-k,k}(P_2),F_{3-2j-k}^{j-k,k-1},F_{2-2j-k}^{-1+j-k,k}\right),
\end{aligned}
\end{equation}
\begin{equation}
\begin{aligned}
P_2
:=&x_2^{1,0}x_3^{-1,0}x_3^{0,-1}x_3^{0,1}x_4^{0,0}
+x_2^{0,1}x_3^{-1,0}x_3^{0,-1}x_3^{1,0}x_4^{0,0}\\
&+x_2^{0,0}x_3^{-1,0}x_3^{0,1}x_3^{1,-1}x_4^{0,0}
+x_2^{0,0}x_3^{-1,1}x_3^{0,-1}x_3^{1,0}x_4^{0,0}\\
&+x_2^{0,0}x_3^{-1,0}x_3^{0,1}x_3^{1,0}x_4^{0,-1}
+x_2^{0,0}x_3^{0,1}x_3^{0,-1}x_3^{1,0}x_4^{-1,0}.
\end{aligned}
\end{equation}
Since $s_{-2j-k}^{j-k,k}(P_2)$ depends on only $x_{3-2j-k}^{j-k,k-1}\in X_0$, $s_{-2j-k}^{j-k,k}(P_2)$ does not change before the mutation at $x_{3-2j-k}^{j-k,k-1}$.
We have
\begin{equation}
\left(s_{-2j-k}^{j-k,k}(P_2),F_{3-2j-k}^{j-k,k-1}\right)
\xrightarrow{\lambda_2}
\left(s_{-2j-k}^{j-k,k}(P_2),F_{3-2j-k}^{j-k,k-1}\right).
\end{equation}
Mutation at $x_{3-2j-k}^{j-k,k-1}$ and calculation of $\left(s_{-2j-k}^{j-k,k}(P_2)\right)'$ from $s_{-2j-k}^{j-k,k}(P_2)$ and $F_{3-2j-k}^{j-k,k-1}$ give
\begin{equation}
\left(s_{-2j-k}^{j-k,k}(P_2),F_{3-2j-k}^{j-k,k-1}\right)
\xrightarrow{\mu_{3-2j-k}^{j-k,k-1}}
\left(\tilde{u}\left(F_{2-2j-k}^{j-k,k}\right),F_{3-2j-k}^{j-k,k-1}\right).
\end{equation}
Since $\tilde{u}\left(F_{2-2j-k}^{j-k,k}\right)$ does not depend on any $x\in X_0$ we have
\begin{equation}
\tilde{u}\left(F_{2-2j-k}^{j-k,k}\right)
\xrightarrow{\lambda_3}
\tilde{u}\left(F_{2-2j-k}^{j-k,k}\right).
\end{equation}
\end{itemize}
The cluster variable $x_{3-2j-k}^{j-k,k}$ does not change by the mutations at $x\in X_0$ and we have
\begin{equation}
\left(x_{3-2j-k}^{j-k,k},F_{3-2j-k}^{j-k,k}\right)
\to \left(x_{3-2j-k}^{j-k,k},\tilde{u}\left(F_{2-2j-k}^{j-k,k}\right)\right)
\end{equation}
by the mutations of initial seed $t_0$ at each and all $x\in X_0$ just once.
\end{itemize}
Therefore, it holds that
\begin{equation}
\begin{aligned}
t_0
=&\left\{\left.\left(x_{-2j-k}^{j-k,k},F_{-2j-k}^{j-k,k}\right),
\left(x_{i-2j-k}^{j-k,k},F_{i-2j-k}^{j-k,k}\right)
\right|i\in[1,5],j,k\in\mathbb{Z}\right\}\\
\to t_1
=&\left\{\left.\left(x_{6-2(j+2)-(k+1)}^{(j+2)-(k+1),k+1},\tilde{u}\left(F_{5-2(j+2)-(k+1)}^{(j+2)-(k+1),k+1}\right)\right),\right.\right.\\
&\quad\left.\left.\left(x_{i-2j-k}^{j-k,k},\tilde{u}\left(F_{-1+i-2j-k}^{j-k,k}\right)\right)
\right|i\in[1,5],j,k\in\mathbb{Z}\right\},
\end{aligned}
\end{equation}
which does not depend on the order of mutations.
\qed
\end{proof}

Note that the iteration of the mutations $\tilde{\mu}_0$ does not depend on the order of the mutations.

\begin{thm}\label{dBKP2}
We defined the seed $t_n$ by $t_n=\tilde{\mu}_{n-1}(t_{n-1})$.
When we put $x_{n-1-2j-k}^{j-k,k}\to x_{n-2j-k}^{1+j-k,1+k}$, it holds that
\begin{equation}\label{dBKP1 seed2}
t_n=\left\{\left.\left(x_{n+i-2j-k}^{j-k,k},\tilde{u}^n\left(F_{i-2j-k}^{j-k,k}\right)\right)\right|i\in[0,5],j,k\in\mathbb{Z}\right\}.
\end{equation}
All cluster variables $x_n^{m,l}\ (n,m,l\in\mathbb{Z})$ satisfy
\begin{equation}\label{eq:dBKP}
x_{n+1}^{m+1,l+1}x_{n}^{m,l}
=x_{n}^{m+1,l+1}x_{n+1}^{m,l}
+x_{n+1}^{m+1,l}x_{n}^{m,l+1}
+x_{n+1}^{m,l+1}x_{n}^{m+1,l}.
\end{equation}
\end{thm}

\begin{proof}
Eqs. \eqref{dBKP1 seed2} and \eqref{eq:dBKP} are obtained by the shift of the suffixes in Eqs. \eqref{dBKP1 seed1} and \eqref{eq}.
Therefore, proof of Theorem \ref{dBKP2} is the same as that of Proposition \ref{dBKP1}.
Equation \eqref{eq:dBKP} is obtained by the definition of the new cluster variable
\begin{equation}
x_{1+n-2j-k}^{1+j-k,1+k}=\tilde{u}^n\left(F_{-2j-k}^{j-k,k}\right)/x_{n-2j-k}^{j-k,k}
\end{equation}
by the mutation of the seed $t_n$ at $x_{n-2j-k}^{j-k,k}$.
\qed
\end{proof}

Equation \eqref{eq:dBKP} is called the discrete BKP equation \cite{BKP}.
The discrete BKP equation is the relation among 8 point of Figure \ref{fig:dBKP}.
\begin{figure}
\begin{center}
\includegraphics[width=3cm]{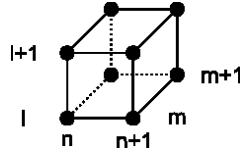}
\end{center}
\caption{The discrete BKP equation}
\label{fig:dBKP}
\end{figure}
By Proposition \ref{dBKP1}, it holds
\begin{equation}
\tilde{\mu}_0(t_0)=\left\{\left.\left(x_{1+i-2j-k}^{j-k,k},\tilde{u}\left(F_{i-2j-k}^{j-k,k}\right)\right)\right|i\in[0,5],j,k\in\mathbb{Z}\right\}.
\end{equation}
The initial seed $t_0$ satisfies the condition which is similar to the one in the Definition \ref{period}.
Therefore, the initial seed $t_0$ is regarded as a generalization of period-1 seed.
By Theorem \ref{Laurent},  all $x_n^{m,l}$ are Laurent polynomials of the initial values in $X_0$.

\subsection{Another seed of the discrete BKP equation}

Let us take
\begin{equation}\label{odBKP seed}
t_0=\left\{\left.\left(x_{i-4j-2k}^{j,k},F_{i-4j-2k}^{j,k}\right)\right|i\in[0,6],j,k\in\mathbb{Z}\right\}
\end{equation}
as an initial seed, where exchange polynomials are defined as
\begin{equation}\label{odBKP polynomial}
\begin{aligned}
F_0^{0,0}
&=x_0^{1,1}x_1^{0,0}
+x_0^{1,0}x_1^{0,1}
+x_0^{0,1}x_1^{1,0},\\
F_1^{0,0}
&=x_0^{1,0}x_1^{0,1}x_2^{0,0}
+x_0^{0,1}x_1^{1,0}x_2^{0,0}
+x_0^{0,0}x_1^{1,0}x_2^{0,1}
+x_0^{0,0}x_1^{0,1}x_2^{1,0},\\
F_2^{0,0}
&=x_1^{1,0}x_2^{0,1}x_2^{0,-1}x_3^{0,0}
+x_1^{0,1}x_2^{0,-1}x_2^{1,0}x_3^{0,0}\\
&\quad+x_1^{0,0}x_2^{0,1}x_2^{1,-1}x_3^{0,0}
+x_1^{0,0}x_2^{0,-1}x_2^{1,0}x_3^{0,1}
+x_1^{0,0}x_2^{0,1}x_2^{1,0}x_3^{0,-1},\\
F_3^{0,0}
&=x_2^{0,-1}x_3^{0,1}x_4^{0,0}
+x_2^{0,1}x_3^{0,-1}x_4^{0,0}
+x_2^{0,0}x_3^{0,-1}x_4^{0,1}
+x_2^{0,0}x_3^{0,1}x_4^{0,-1},\\
F_4^{0,0}
&=x_5^{-1,0}x_4^{0,-1}x_4^{0,1}x_3^{0,0}
+x_5^{0,-1}x_4^{0,1}x_4^{-1,0}x_3^{0,0}\\
&\quad+x_5^{0,0}x_4^{0,-1}x_4^{-1,1}x_3^{0,0}
+x_5^{0,0}x_4^{0,1}x_4^{-1,0}x_3^{0,-1}
+x_5^{0,0}x_4^{0,-1}x_4^{-1,0}x_3^{0,1},\\
F_5^{0,0}
&=x_6^{-1,0}x_5^{0,-1}x_4^{0,0}
+x_6^{0,-1}x_5^{-1,0}x_4^{0,0}
+x_6^{0,0}x_5^{-1,0}x_4^{0,-1}
+x_6^{0,0}x_5^{0,-1}x_4^{-1,0},\\
F_6^{0,0}
&=x_6^{-1,-1}x_5^{0,0}
+x_6^{-1,0}x_5^{0,-1}
+x_6^{0,-1}x_5^{-1,0},
\end{aligned}
\end{equation}
and
\begin{equation}
F_{i-4j-2k}^{j,k}=s_{-4j-2k}^{j,k}\left(F_i^{0,0}\right).
\end{equation}
Now we define $\tilde{u}=s_1^{0,0} and$ the set of cluster variables $X_i$ as
\begin{equation}
X_i:=\left\{\left.x_{i-4j-2k}^{j,k}\right|j,k\in\mathbb{Z}\right\}.
\end{equation}

\begin{thm}\label{odBKP}
We define the seed $t_n$ as the seed obtained by the iteration of the mutations of the seed $t_{n-1}$ at each and all $x\in X_{n-1}$ just once.
When we put $x_{n-1-4j-2k}^{j,k}\to x_{n-4j-2k}^{1+j,1+k}$, we find
\begin{equation}
t_n=\left\{\left.\left(x_{n+i-4j-2k}^{j,k},\tilde{u}^n\left(F_{i-4j-2k}^{j,k}\right)\right)\right|i\in[0,6],j,k\in\mathbb{Z}\right\},
\end{equation}
which does not depend on the order of mutations.
All cluster variables $x_n^{m,l}\ (n,m,l\in\mathbb{Z})$ satisfy the discrete BKP equation \eqref{eq:dBKP}.
\end{thm}

\begin{proof}
Proof of Theorem \ref{odBKP} can be done in a similar manner to those of Proposition \ref{dBKP1} and Theorem \ref{dBKP2}.
\end{proof}

\section{Several difference equations associated with reductions of the seed}

In this section, we show that several seeds which give the difference equations obtained by imposing constraints on the BKP equation and discuss the relation to the reductions for initial seeds and polynomials in Laurent phenomenon algebra.

\subsection{2-dimensional difference equation (1)}

We take
\begin{equation}\label{2d1 seed}
t_0=\left\{\left.\left(x_{i-2j}^j,F_{i-2j}^j\right)\right|i\in[0,5],j\in\mathbb{Z}\right\}
\end{equation}
as an initial seed, where exchange polynomials are defined as
\begin{equation}\label{2d1 polynomial}
\begin{aligned}
F_{0}^{0}
&=x_1^2x_1^0
+x_0^1x_2^1
+\left(x_1^1\right)^2,\\
F_{1}^{0}
&=x_0^1x_2^1x_2^0
+\left(x_1^1\right)^2x_2^0
+x_0^0x_1^1x_3^1
+x_0^0\left(x_2^1\right)^2,\\
F_{2}^{0}
&=x_1^1x_2^{-1}x_3^1x_3^0
+\left(x_2^1\right)^2x_2^{-1}x_3^0
+x_1^0\left(x_3^0\right)^2x_2^1
+x_1^0x_2^{-1}\left(x_3^1\right)^2
+x_1^0x_3^1x_2^1x_3^{-1},\\
F_{3}^{0}
&=x_4^{-1}x_3^1x_2^{-1}x_2^0
+\left(x_3^{-1}\right)^2x_3^1x_2^0
+x_4^0\left(x_2^0\right)^2x_3^{-1}
+x_4^0x_3^1\left(x_2^{-1}\right)^2
+x_4^0x_2^{-1}x_3^{-1}x_2^1,\\
F_{4}^{0}
&=x_5^{-1}x_3^{-1}x_3^0
+\left(x_4^{-1}\right)^2x_3^0
+x_5^0x_4^{-1}x_2^{-1}
+x_5^0\left(x_3^{-1}\right)^2,\\
F_{5}^{0}
&=x_4^{-2}x_4^0
+x_5^{-1}x_3^{-1}
+\left(x_4^{-1}\right)^2,
\end{aligned}
\end{equation}
and
\begin{equation}
F_{i-2j}^j=s_{-2j}^j\left(F_i^0\right).
\end{equation}
We define  $\tilde{u}:=s_1^0$ and the set of cluster variables $X_i$ by
\begin{equation}
X_i:=\left\{\left.x_{i-2j}^j\right|j\in\mathbb{Z}\right\}.
\end{equation}

\begin{prop}\label{2d1}
We define the seed $t_n$ as the seed obtained by the iteration of the mutations of the seed $t_{n-1}$ at each and all $x\in X_{n-1}$ just once.
When we put $x_{n-1-2j}^j\to x_{n+1-2j}^{2+j}$, we have
\begin{equation}
t_n=\left\{\left.\left(x_{n+i-2j}^j,\tilde{u}^n\left(F_{i-2j}^j\right)\right)\right|i\in[0,5],j\in\mathbb{Z}\right\},
\end{equation}
which does not depend on the order of mutations.
All cluster variables $x_n^m\ (n,m\in\mathbb{Z})$ satisfy
\begin{equation}\label{eq:2d1}
x_{n+2}^{m+2}x_{n}^{m}
=x_{n+1}^{m+2}x_{n+1}^{m}
+\left(x_{n+1}^{m+1}\right)^2
+x_{n+2}^{m+1}x_{n}^{m+1}.
\end{equation}
\end{prop}

2-dimensional equation \eqref{eq:2d1} is the relation among 7 points shown in Fig.~\ref{fig:2d1}.
\begin{figure}
\begin{center}
\includegraphics[width=3cm]{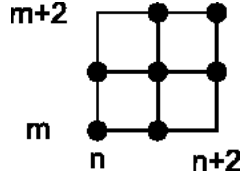}
\end{center}
\caption{2-dimensional difference equation (1)}
\label{fig:2d1}
\end{figure}
Note that 2-dimensional difference equation \eqref{eq:2d1} is obtained from the discrete BKP equation \eqref{eq:dBKP} by imposing the reduction condition
\begin{equation}\label{reduction1}
x_n^{m,l+1}=x_{n+1}^{m+1,l},\quad x_n^m:=x_n^{m,0}.
\end{equation}
The initial seed \eqref{2d1 seed} and the exchange polynomials \eqref{2d1 polynomial} are obtained from the initial seed \eqref{dBKP seed} and the exchange polynomials \eqref{dBKP polynomial} by imposing the same condition \eqref{reduction1}.
Therefore, the reduction of the seed corresponds to the reduction of the difference equation.
Proof of Proposition \ref{2d1} is the same as those of Proposition \ref{dBKP1} and Theorem \ref{dBKP2}.

\subsection{Somos-6}

We impose the reduction condition
\begin{equation}\label{reduction2}
x_n^{m+1}=x_{n+2}^m,\quad x_n:=x_n^0
\end{equation}
to the initial seed \eqref{2d1 seed} and exchange polynomials \eqref{2d1 polynomial}.
Then we obtain the seed and exchange polynomials
\begin{equation}
t_0=\{\left.(x_i,F_i)\right|i\in[0,5]\},
\end{equation}
\begin{equation}
\begin{aligned}
F_{0}
&=x_5x_1
+x_2x_4
+(x_3)^2,\\
F_{1}
&=(x_2)^2x_4
+(x_3)^2x_2
+x_0x_3x_5
+x_0(x_4)^2,\\
F_{2}
&=(x_3)^2x_0x_5
+(x_4)^2x_0x_3
+x_1(x_3)^2x_4
+x_1x_0(x_5)^2
+(x_1)^2x_5x_4,\\
F_{3}
&=(x_2)^2x_5x_0
+(x_1)^2x_5x_2
+x_4(x_2)^2x_1
+x_4x_5(x_0)^2
+(x_4)^2x_0x_1,\\
F_{4}
&=(x_3)^2x_1
+(x_2)^2x_3
+x_5x_2x_0
+x_5(x_1)^2,\\
F_{5}
&=x_0x_4
+x_3x_1
+(x_2)^2.
\end{aligned}
\end{equation}
This seed is a period-1 seed.
This seed has already been obtained in \cite{LP}.
\begin{prop}{\rm\cite{LP}}\label{Somos-6}
We define the seed $t_n$ by $t_n=\mu_{n-1}(t_{n-1})$.
When we put $x_{n-1}\to x_{n+5}$, it holds that
\begin{equation}
t_n=\left\{\left.\left(x_{n+i},u^n\left(F_i\right)\right)\right|i\in[0,5]\right\}.
\end{equation}
All cluster variables $x_n\ (n\in\mathbb{Z})$ satisfy
\begin{equation}\label{eq:Somos-6}
x_{n+6}x_{n}
=x_{n+5}x_{n+1}
+(x_{n+3})^2
+x_{n+4}x_{n+2}.
\end{equation}
\end{prop}

Difference equation \eqref{eq:Somos-6} is called the Somos-6 \cite{Somos}.
The Somos-6 \eqref{eq:Somos-6} is obtained from 2-dimensional difference equation \eqref{eq:2d1} by imposing the reduction condition \eqref{reduction2}.

\subsection{2-dimensional difference equation (2)}

We take
\begin{equation}\label{2d2 seed}
t_0=\left\{\left.\left(x_{i-4j}^j,F_{i-4j}^j\right)\right|i\in[0,6],j\in\mathbb{Z}\right\}
\end{equation}
as an initial seed, where exchange polynomials are defined as
\begin{equation}\label{2d2 polynomial}
\begin{aligned}
F_0^0
&=x_2^1x_1^0
+x_0^1x_3^0
+x_2^0x_1^1,\\
F_1^0
&=x_0^1x_3^0x_2^0
+\left(x_2^0\right)^2x_1^1
+x_0^0x_1^1x_4^0
+x_0^0x_3^0x_2^1,\\
F_2^0
&=x_1^1x_4^0x_0^0x_3^0
+\left(x_3^0\right)^2x_0^0x_2^1
+x_1^0x_4^0x_0^1x_3^0
+x_1^0x_0^0x_2^1x_5^0
+\left(x_1^0\right)^2x_4^0x_2^1,\\
F_3^0
&=x_0^0x_5^0x_4^0
+\left(x_4^0\right)^2x_1^0
+x_2^0x_1^0x_6^0
+\left(x_2^0\right)^2x_5^0,\\
F_4^0
&=x_5^{-1}x_2^0x_6^0x_3^0
+\left(x_3^0\right)^2x_6^0x_4^{-1}
+x_5^0x_2^0x_6^{-1}x_3^0
+x_5^0x_6^0x_4^{-1}x_1^0
+\left(x_5^0\right)^2x_2^0x_4^{-1},\\
F_5^0
&=x_6^{-1}x_3^0x_4^0
+\left(x_4^0\right)^2x_5^{-1}
+x_6^0x_5^{-1}x_2^0
+x_6^0x_3^0x_4^{-1},\\
F_6^0
&=x_4^{-1}x_5^0
+x_6^{-1}x_3^0
+x_4^0x_5^{-1},
\end{aligned}
\end{equation}
and
\begin{equation}
F_{i-4j}^j=s_{-4j}^j\left(F_i^0\right).
\end{equation}
We define $\tilde{u}=s_1^0$ and the set of cluster variables $X_i$ as
\begin{equation}
X_i:=\left\{\left.x_{i-4j}^j\right|j\in\mathbb{Z}\right\}.
\end{equation}

\begin{prop}\label{2d2}
We define the seed $t_n$ as the seed obtained by the iteration of the mutations of the seed $t_{n-1}$ at each and all $x\in X_{n-1}$ just once.
Putting $x_{n-1-4j}^j\to x_{n+2-4j}^{j+1}$, we find
\begin{equation}
t_n=\left\{\left.\left(x_{n+i-4j}^j,\tilde{u}^n\left(F_{i-4j}^j\right)\right)\right|i\in[0,6],j\in\mathbb{Z}\right\},
\end{equation}
which does not depend on the order of mutations.
All cluster variables $x_n^m\ (n,m\in\mathbb{Z})$ satisfy
\begin{equation}\label{eq:2d2}
x_{n+3}^{m+1}x_n^m
=x_{n+2}^{m+1}x_{n+1}^m
+x_{n+1}^{m+1}x_{n+2}^m
+x_{n+3}^mx_n^{m+1}.
\end{equation}
\end{prop}

2-dimensional equation \eqref{eq:2d2} is the relation among 8 points shown in Fig.~\ref{fig:2d2}.
\begin{figure}
\begin{center}
\includegraphics[width=3cm]{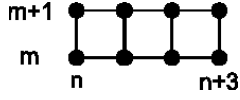}
\end{center}
\caption{2-dimensional difference equation (2)}
\label{fig:2d2}
\end{figure}

Note that 2-dimensional difference equation \eqref{eq:2d2} is obtained from the discrete BKP equation \eqref{eq:dBKP} by imposing the reduction condition
\begin{equation}\label{reduction3}
x_n^{m,l+1}=x_{n+2}^{m,l},\quad x_n^m:=x_n^{m,0}.
\end{equation}
The initial seed \eqref{2d2 seed} and the exchange polynomials \eqref{2d2 polynomial} are obtained from the initial seed \eqref{odBKP seed} and the exchange polynomials \eqref{odBKP polynomial} by imposing the same condition \eqref{reduction3}.

\subsection{Somos-7}

We take
\begin{equation}\label{Somos-7 seed}
t_0=\{\left.(x_i,F_i)\right|i\in[0,6]\}
\end{equation}
as an initial seed, where exchange polynomials are defined as
\begin{equation}\label{Somos-7 polynomial}
\begin{aligned}
F_0
&=x_6x_1
+x_4x_3
+x_2x_5\\
F_1
&=x_4x_3x_2
+(x_2)^2x_5
+x_0x_5x_4
+x_0x_3x_6\\
F_2
&=x_5x_4x_0x_3
+(x_3)^2x_0x_6
+x_1(x_4)^2x_3
+x_1x_0x_6x_5
+(x_1)^2x_4x_6\\
F_3
&=x_0x_5x_4
+(x_4)^2x_1
+x_2x_1x_6
+(x_2)^2x_5\\
F_4
&=x_1x_2x_6x_3
+(x_3)^2x_6x_0
+x_5(x_2)^2x_3
+x_5x_6x_0x_1
+(x_5)^2x_2x_0\\
F_5
&=x_2x_3x_4
+(x_4)^2x_1
+x_6x_1x_2
+x_6x_3x_0\\
F_6
&=x_0x_5
+x_2x_3
+x_4x_1.
\end{aligned}
\end{equation}

\begin{prop}\label{Somos-7}
We define the seed $t_n$ by $t_n=\mu_{n-1}(t_{n-1})$.
We put $x_{n-1}\to x_{n+6}$ and we have
\begin{equation}
t_n=\left\{\left.\left(x_{n+i},u^n\left(F_i\right)\right)\right|i\in[0,6]\right\}.
\end{equation}
All cluster variables $x_n\ (n\in\mathbb{Z})$ satisfy
\begin{equation}\label{eq:Somos-7}
x_{n+7}x_n
=x_{n+6}x_{n+1}
+x_{n+5}x_{n+2}
+x_{n+3}x_{n+4}.
\end{equation}
\end{prop}

Difference equation \eqref{eq:Somos-7} is called the Somos-7 \cite{Somos}.
Note that the Somos-7 \eqref{eq:Somos-7} is obtained from the 2-dimensional difference equation \eqref{eq:2d2} by imposing the reduction condition
\begin{equation}\label{reduction4}
x_n^{m+1}=x_{n+4}^m,\quad x_n:=x_n^0.
\end{equation}
The initial seed \eqref{Somos-7 seed} and the exchange polynomials \eqref{Somos-7 polynomial} are obtained from the initial seed \eqref{2d2 seed} and the exchange polynomials \eqref{2d2 polynomial} by imposing the same condition \eqref{reduction4}.

\section{Conclusion}
We have shown that cluster variables can satisfy the discrete BKP equation, the 2-dimensional difference equations of its reductions, and Somos-7, if we take appropriate initial seeds in Laurent phenomenon algebras.
These initial seeds are obtained from reductions of the seed of the discrete BKP equation.
It is known that cluster variables of suitable cluster algebras can satisfy the bilinear form of some $q$-discrete Painlev\'e equations \cite{dPanleve}, when the initial seed includes appropriate periodic quivers \cite{quiver2}.
However, we have not obtained the $q$-discrete Painlev\'e equations of type $A_4^{(1)}$ and $E_n^{(1)}$ from cluster algebras.
To clarify the relation between these equations and Laurent phenomenon algebras is one of the problems we wish to address in the future.

\end{document}